\documentclass[12pt, a4paper]{amsart}

\usepackage{amsmath,amssymb,amsfonts}
\usepackage{bbold}
\usepackage{epic,eepic}
\usepackage{enumerate}

\usepackage{amsthm}
\usepackage[all]{xy}
\usepackage{caption}
\usepackage{stmaryrd} 
\usepackage[dvipdfmx]{graphicx,color} 
\usepackage{tikz}
\usetikzlibrary{trees}

\usepackage{geometry}
\usepackage{listings} 

\theoremstyle{plain}
\newtheorem{thm}{Theorem}[section]
\newtheorem{prop}[thm]{Proposition}
\newtheorem{lem}[thm]{Lemma}
\newtheorem{cor}[thm]{Corollary}

\theoremstyle{definition}
\newtheorem{defn}[thm]{Definition}
\newtheorem{exmp}[thm]{Example}

\newtheorem{rem}[thm]{Remark}
\newtheorem{asm}[thm]{Assumption}

\numberwithin{figure}{section}

\numberwithin{table}{section}

\newcommand{\lspace} {
  \vspace{0.8\baselineskip}
}

\newcommand{\single}[1]{
  \langle  #1 \rangle 
}

\newcommand{\abs}[1]{
  \lvert  #1 \rvert
}

\newcommand{\Nat}{\mathrm{Nat}}

\newdir{ >}{{}*!/-5pt/@{>}}

\definecolor{arrowred}{rgb}{0,0,0} 
\newcommand{\newword}[1]{\textbf{\textit{#1}}}

\newcommand{\textred}[1]{#1}
\newcommand{\textblue}[1]{#1}

\numberwithin{equation}{section}

\newcommand{\ProbB}{\bar{B}}
\newcommand{\ProbC}{\bar{C}}
\newcommand{\ProbX}{\bar{X}}
\newcommand{\ProbY}{\bar{Y}}

\newcommand{\ProbNX}{\bar{X}}
\newcommand{\ProbNY}{\bar{Y}}

\newcommand{\Prob}{\mathbf{Prob}}

\newcommand{\Meas}{\mathbf{Meas}}

\newcommand{\Set}{\mathbf{Set}}

\mathchardef\mhyphen="2D  

\title {A Binomial Asset Pricing Model in a Categorical Setting}

\thanks{This work was supported by JSPS KAKENHI Grant Number 18K01551.}

\author[T. Adachi, K. Nakajima and Y. Ryu]{Takanori Adachi, Katsushi Nakajima and Yoshihiro Ryu}
\address{Graduate School of Management,
         Tokyo Metropolitan University,
         1-4-1 Marunouchi, Chiyoda-ku, Tokyo 100-0005, Japan}
\email{Takanori Adachi <tadachi@tmu.ac.jp>}

\address{College of International Management,
         Ritsumeikan Asia Pacific University,
         1-1 Jumonjibaru, Beppu, Oita, 874-8577 Japan}
\email{Katsushi Nakajima <knakaji@apu.ac.jp>}

\address{Department of Mathematical Sciences,
         Ritsumeikan University,
         1-1-1 Nojihigashi, Kusatsu, Shiga, 525-8577 Japan}
\email{Yoshihiro Ryu <iti2san@gmail.com>}

\date{\today}



\begin{document}

\maketitle

\begin{abstract}
Adachi and Ryu introduced a category $\Prob$ of probability spaces whose objects are 
all probability spaces
and whose arrows correspond to measurable functions
 satisfying an absolutely continuous requirement
in \cite{AR_2019}.
In this paper,
we develop a binomial asset pricing model based on $\Prob$.
We introduce generalized filtrations with which we can represent situations
such as
some agents forget information at some specific time.
We investigate the
valuations of financial claims along this type of non-standard filtrations. 

\end{abstract}



\section{Introduction}
\label{sec:intro}

Adachi and Ryu introduced the category
$\Prob$
as an 
\textit{adequate}
candidate of 
the category of probability spaces with 
\textit{good} arrows.
They show the existence  of
the conditional expectation functor 
from 
$\Prob$
to
$\Set$,
which is a natural generalization of
the classical notion of conditional expectation
(\cite{AR_2019}).,

In this paper,
we develop a binomial asset pricing model based on the category
$\Prob$.
Generalized filtrations defined in this setting change not only $\sigma$-algebras but also
probability measures and even underlying sets throughout time.
We introduce a few types of generalized filtrations.
Each of them represents a subjective filtration of an agent.
In other words,
each agent has not only her subjective probability measure but also her own subjective filtration.
For example, some filtration represents the situation in which she forgets the information generated at a specific time.
This paper investigate the valuations of financial claims along these non-standard filtrations. 

\lspace

First, in 
Section \ref{sec:GenFilt},
we review the concept of categorical probability theory
and introduce generalized filtrations and adapted processes and martingales along them.
In this setting, our probability spaces are changing as time goes on.
For example, we may have a bigger underlying set in future than that in past.
This case allows us to have unknown future elementary events.
Section \ref{sec:BAPM} is the heart of this paper in which
we develop a concrete binomial asset pricing model
and investigate a few generalized filtrations
and possibility of valuations along them.
We also provide a complete form of a replication strategy making the valuation possible.

\section{Generalized Filtrations}
\label{sec:GenFilt}

In this section,
we introduce some basic concepts of categorical probability theory
which were mainly introduced in
\cite{AR_2019}
as a preparation for 
Section \ref{sec:BAPM}.

\lspace

Let
$
\bar{X} =
(X, \Sigma_X, \mathbb{P}_X)
$,
$
\bar{Y} =
(Y, \Sigma_Y, \mathbb{P}_Y)
$
and
$
\bar{Z} =
(Z, \Sigma_Z, \mathbb{P}_Z)
$
be
probability spaces
throughout this paper.

\begin{defn}{[Null-preserving functions 
\cite{AR_2019}
]}
A measurable function
$
f
  :
\ProbY
  \to
\ProbX
$
is called
\newword{null-preserving}
if
$
f^{-1}(A)
\in \mathcal{N}_Y
$
for every
$
A \in \mathcal{N}_X
$,
where
$
\mathcal{N}_X
  :=
\mathbb{P}_X^{-1}(0)
  \subset
\Sigma_X
$
and
$
\mathcal{N}_Y
  :=
\mathbb{P}_Y^{-1}(0)
  \subset
\Sigma_Y
$.
\end{defn}

\begin{defn}{[Category $\Prob$
\cite{AR_2019}
]}
\label{defn:cat_prob}
A category
$
\Prob
$
is the category whose objects are
all probability spaces and 
the set of arrows between them are defined by
\begin{align*}
\Prob(
  \ProbX,
&
  \ProbY
)
  :=
\{
f^-
  \mid
f :
  \ProbY
\to
  \ProbX
\textrm{ is a null-preserving function.}
\},
\end{align*}
where
$f^-$
is a symbol
corresponding uniquely to
a function $f$.

We write 
$Id_X$
for an identity \textit{measurable} function 
from
$\ProbX$
to
$\ProbX$,
while writing
$id_X$
for an identity function
from
$X$
to
$X$.
Therefore,
the identity arrow of
a
$\Prob$-object
$\ProbX$
is 
$Id_X^-$.

\end{defn}

\begin{defn}{[Generalized Filtrations]}
Let
$
\textred{
\mathcal{T}
}
$
be a fixed small category
which we sometimes call the
\newword{time domain}.
A
\newword{$\mathcal{T}$-filtration}
is a functor
$
F :
 \mathcal{T}
\to
 \Prob
$.
\end{defn}

\begin{figure}[h]
\begin{equation*}
\xymatrix@C=10 pt@R=7 pt{
	\mathcal{T}
		\ar @{->}_{\textred{F}} [dd]
&
	t_0
		\ar @{->}_{i_0} [rr]
		\ar @/^2pc/^{ i_1 \circ i_0 } [rrrr]
&&
	t_1
		\ar @{->}_{i_1} [rr]
&&
	t_2
		\ar @{->}_{i_2} [rr]
&&
	\dots
\\\\
   \Prob
&
   F t_0
		\ar @{->}^{Fi_0} [rr]
		\ar @/_2pc/_{
		  \textblue{
			F(i_1 \circ i_0) = Fi_1 \circ Fi_0
		  }
		} [rrrr]
&&
   F t_1
		\ar @{->}^{Fi_1} [rr]
&&
   F t_2
		\ar @{->}^{Fi_2} [rr]
&&
	\dots
}
\end{equation*}
\caption{$\mathcal{T}$-filtration}
\label{fig:t_filtration}
\end{figure}
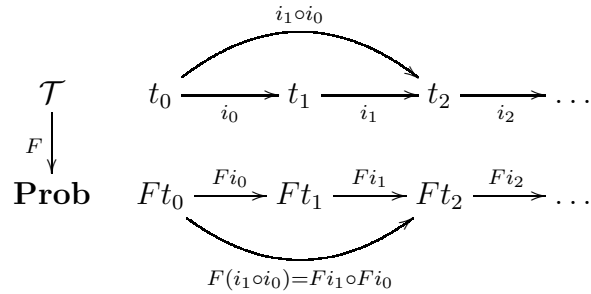

When we say filtrations in the classical setting, 
we keep using a same underlying set
$\Omega$
throughout time.
This situation can be represented by the following diagram.

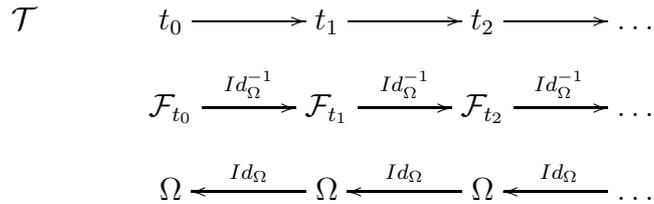
\begin{equation*}
\xymatrix@C=35 pt@R=15 pt{
\mathcal{T}
&
	t_0
       \ar @{->} [r]
&
	t_1
       \ar @{->} [r]
&
	t_2
       \ar @{->} [r]
&
	\dots
\\
&
	\textred{\mathcal{F}_{t_0}}
       \ar @{->}^{Id_{\Omega}^{-1}} [r]
&
	\textred{\mathcal{F}_{t_1}}
       \ar @{->}^{Id_{\Omega}^{-1}} [r]
&
	\textred{\mathcal{F}_{t_2}}
       \ar @{->}^{Id_{\Omega}^{-1}} [r]
&
	\dots
\\
&
	\Omega
&
	\Omega
       \ar @{->}_{Id_{\Omega}} [l]
&
	\Omega
       \ar @{->}_{Id_{\Omega}} [l]
&
	\dots
       \ar @{->}_{Id_{\Omega}} [l]
}
\end{equation*}

However, in our new setting, the filtration can change not only 
$\sigma$-fields but also
probability measures and underlying sets as the following diagram shows.

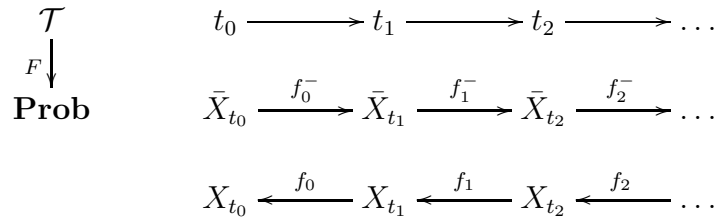
\begin{equation*}
\xymatrix@C=35 pt@R=15 pt{
\mathcal{T}
       \ar @{->}_{\textred{F}} [d]
&
	t_0
       \ar @{->} [r]
&
	t_1
       \ar @{->} [r]
&
	t_2
       \ar @{->} [r]
&
	\dots
\\
\textred{\Prob}
&
	\ProbX_{t_0}
       \ar @{->}^{\textred{f^-_{0}}} [r]
&
	\ProbX_{t_1}
       \ar @{->}^{\textred{f^-_{1}}} [r]
&
	\ProbX_{t_2}
       \ar @{->}^{\textred{f^-_{2}}} [r]
&
	\dots
\\
&
	X_{t_0}
&
	X_{t_1}
       \ar @{->}_{f_{0}} [l]
&
	X_{t_2}
       \ar @{->}_{f_{1}} [l]
&
	\dots
       \ar @{->}_{f_{2}} [l]
}
\end{equation*}

One of the implications of this generalization is 
that we can think possibly distorted filtrations
by using adequate null-preserving function
$f_t$.

Actually,
the biggest aim of this paper is to investigate this kind of non-standard filtrations
by using, as a first example, a simple binomial asset pricing model.

\lspace

Let
$
\Meas
$
be the category of measurable spaces that consists of all measurable spaces as its objects
and all measurable functions between them as arrows.
\begin{equation}
U : \Prob \to \Meas
\end{equation}
is a forgetful functor that is it maps
a object $\ProbX$
to
$
(X, \Sigma_X)
$
by dropping its probability measure,
and an arrow
$
f : \ProbX \to \ProbY
$
to
$
f : (X, \Sigma_X) \to (Y, \Sigma_Y)
$.

Later, we will investigate a modification of a given filtration
$F$
to
another filtration
$G$
such that 
$
U \circ F
	=
U \circ G
$,
that is, the situation when they share their measurable space nature.

\lspace

Before going into our concrete example,
we will define adapted processes and martingales over this generalized filtrations.

Let
$F$
be a
fixed
$\mathcal{T}$-filtration
throughout this section.

\begin{defn}{[$F$-Adapted Processes]}
\label{defn:adaPro}
\newword{An $F$-adapted process}
is a collection of natural transformations
\begin{equation}
\tau
  :=
\{
   \tau_s
      :
   \mathcal{T}(s, -)
      \dot{\to}
   L \circ F
\}_{s \in Obj(\mathcal{T})}
\end{equation}

\end{defn}

For a 
$\Prob$-arrow
$
\varphi
  :
\ProbX
  \to
\ProbY
$,
there exists a measurable function
$
f
  :
Y
  \to
X
$
such that
$
\varphi = f^-
$
by its definition.
We write
$
\varphi^+
$
for this $f$.
That is,
$
(\varphi^+)^- = \varphi
$.

Now
Let
$\tau$
be an 
$F$-adapted
process
and
$
i : s \to t
$
be a
$\mathcal{T}$-arrow.
Then, we have the following commutative diagram.
\begin{equation*}
\xymatrix@C=20 pt@R=30 pt{
&
   Id_s
     \ar @{}_{ \mathrel{ \rotatebox[origin=c]{-45}{$\in$} } } @<+6pt> [rd]
     \ar @{|->} [rrr]
     \ar @{|->} [ddd]
&&&
   \tau_{s,s}(Id_s)
     \ar @{}_{ \mathrel{ \rotatebox[origin=c]{-135}{$\in$} } } @<+12pt> [ld]
     \ar @{|->} [ddd]
\\
s
     \ar @{->}_{i} [d]
&&
   \mathcal{T}(s, s)
      \ar @{->}^{\tau_{s,s}} [r]
      \ar @{->}_{\mathcal{T}(s, i)} [d]
&
   L(Fs)
      \ar @{->}^{L(Fi)} [d]
&
\\
t
&&
   \mathcal{T}(s, t)
      \ar @{->}_{\tau_{s,t}} [r]
&
   L(Ft)
&
\\
&
   i
     \ar @{}_{ \mathrel{ \rotatebox[origin=c]{45}{$\in$} } } @<+12pt> [ru]
     \ar @{|->} [rr]
&&
   \tau_{s,t}(i)
     \ar @{}_-{=} @<+6pt> [r]
&
   L(Fi)(\tau_{s,s}(Id_s))
     \ar @{}_{ \mathrel{ \rotatebox[origin=c]{135}{$\in$} } } @<+6pt> [lu]
}
\end{equation*}
For 
$
s \in Obj(\mathcal{T})
$
pick a random variable
$v_s$
satisfying
$
[
  v_s
]_{\sim_{\mathbb{P}_{Fs}}}
  =
\tau_{s,s}(Id_s)
$.
Then, we have
\begin{equation}
\tau_{s,t}(i)
  =
[
   v_s \circ (Fi)^+
]_{\sim_{\mathbb{P}_{Ft}}}.
\end{equation}
That is,
$
\tau_{s,t}(i)
$
is
$(Fi)$-measurable.
%


\begin{prop}
Let
$
AP(F)
$
be the set of all
$F$-adapted processes.
Then,
\begin{equation}
\label{eq:APeq}
AP(F)
       \cong
\prod_{t \in Obj(\mathcal{T})}
  L(Ft) .
\end{equation}
\end{prop}
\begin{proof}
By Yoneda Lemma, we have
for 
$
t
 \in
Obj(\mathcal{T})
$,
\begin{equation}
y_t :
\Nat(\mathcal{T}(t, -), L \circ F)
  \cong
(L \circ F)t .
\end{equation}
Then,
$
\prod_{t \in Obj(\mathcal{T})}
  y_t
$
is an isomorphism denoting
(\ref{eq:APeq}).

\end{proof}

For
$
\textblue{
x
}
	\in
AP(F)
$,
we sometimes write
\begin{equation}
x
	=
\{
	\textred{x_t}
\}_{
	t \in Obj(\mathcal{T})
}
\end{equation}
where
\begin{equation}
\textred{
x_t
}
	:=
x(t)
	\in
L(Ft) .
\end{equation}

\begin{rem}
For an arrow
$
i : s \to t
$
in
$\mathcal{T}$,
in general,
$Fs$
and
$Ft$
are different probability spaces.
So we cannot (for example) add two random variables
$x_s \in L^1(Fs)$
and
$x_t \in L^1(Ft)$
whose domains are
$\tilde{Fs}$
and
$\tilde{Ft}$.
\begin{equation*}
\xymatrix@C=30 pt@R=10 pt{
	s
		\ar @{->}^{i} [r]
&
	t
\\
	Fs
		\ar @{->}^{Fi} [r]
&
	Ft
\\
	\tilde{Fs}
		\ar @{->}_{x_s} [dd]
&
	\tilde{Ft}
		\ar @{->}_{(Fi)^+} [l]
		\ar @{->}_{x_t} @<-2pt> [dd]
		\ar @{->}^{\textblue{x_s \circ (Fi)^+}} @<+2pt> [dd]
\\\\
	\mathbb{R}
&
	\mathbb{R}
}
\end{equation*}
In order to import 
$x_s$
into
$L^1(Ft)$,
we take 
$
x_s \circ (Fi)^+
$
as its proxy.
This fact allows us to treat 
$L^1(Ft)$
as a vector space containing all preceding random variables
$x_s \in L^1(Fs)$
with
$s \le t$.

\end{rem}

Next, we go into the definition of martingales.
In order to make it possible, we need a concept of conditional expectations in
the category
$\Prob$
which was introduced in
\cite{AR_2019}.

\begin{thm} {\normalfont{[Conditional Expectation \cite{AR_2019}]}}
\label{thm:cond_E}
Let 
 $f^- : \ProbX \to \ProbY$
be a 
$\Prob$-arrow.
For all 
$v \in \mathcal{L}^1(\ProbNY)$
and
$A \in \Sigma_X$,
there exists
$u \in \mathcal{L}^1(\ProbNX)$
satisfying the following equation.
\begin{equation}
\label{eq:ceGen}
 \int_A u \, d \mathbb{P}_X
 =
 \int_{f^{-1}(A)} v \, d \mathbb{P}_Y  .
\end{equation}
We call 
$u$
a \newword{conditional expectation} along $f^-$
and denote it by
$E^{f^-}(v)$.
\end{thm}

\begin{thm} {\normalfont{[Conditional Expectation Functor \cite{AR_2019}]}}
\label{thm:fun_E}
There exists a functor
$
\textred{\mathcal{E}} : \Prob^{\textred{op}} \to \Set
$
as following:
\begin{equation*}
\xymatrix@C=15 pt@R=30 pt{
    X
  &
    \ProbX
     \ar @{->}_{f^-} [d]
     \ar @{|->}^{\mathcal{E}} [rr]
  &&
    \textred{\mathcal{E}\ProbX}
     \ar @{}^-{:=} @<-6pt> [r]
  &
    L^{1}(\ProbNX)
     \ar @{}^-{\ni} @<-6pt> [r]
  &
    [E^{f^-}(v)]_{\sim_{\mathbb{P}_X}}
\\
    Y
      \ar @{->}^f [u]
  &
    \ProbY
     \ar @{|->}^{\mathcal{E}} [rr]
  &&
    \mathcal{E}\ProbY
     \ar @{}^-{:=} @<-6pt> [r]
     \ar @{->}_{
       \mathcal{E} f^-
     } [u]
  &
    L^{1}(\ProbNY)
     \ar @{}^-{\ni} @<-6pt> [r]
  &
    [v]_{\sim_{\mathbb{P}_Y}}.
     \ar @{|->}_{
       \textred{\mathcal{E} f^-}
     } @<9pt> [u]
}
\end{equation*}
We call
$\mathcal{E}$
a 
\newword{conditional expectation functor}.
\end{thm}

\begin{defn}{[$F$-Martingales]}
\label{defn:martPr1}
Let
$
F : \mathcal{T} \to \Prob
$
be a functor.
An 
\newword{$F$-martingale}
is 
an $F$-adapted process
$
x
  \in
AP(F)
$
such that for every 
$\mathcal{T}$-arrow
$
i : s \to t
$,
\begin{equation}
(\mathcal{E} \circ F)
  i
  (x(t))
=
 x(s) .
\end{equation}
\end{defn}

\begin{figure}[h]
\begin{equation*}
\xymatrix@C=15 pt@R=30 pt{
	s
		\ar @{->}_i [d]
		\ar @{|->}^F [rr]
  &&
    Fs
		\ar @{->}_{Fi} [d]
		\ar @{|->}^{\mathcal{E}} [rr]
  &&
	\mathcal{E}(Fs)
		\ar @{}^-{:=} @<-6pt> [r]
  &
    L^{1}(Fs)
     \ar @{}^-{\ni} @<-6pt> [r]
  &
    \textred{x_s} = [\textblue{E^{Fi}(v)}]_{\sim_{\mathbb{P}_{Fs}}}
\\
    t
		\ar @{|->}^F [rr]
  &&
	Ft
     \ar @{|->}^{\mathcal{E}} [rr]
  &&
	\mathcal{E}(Ft)
		\ar @{}^-{:=} @<-6pt> [r]
		\ar @{->}_{
			\textred{\mathcal{E}(Fi)}
		} [u]
  &
    L^{1}(Ft)
     \ar @{}^-{\ni} @<-6pt> [r]
  &
    \textred{x_t} = [\textblue{v}]_{\sim_{\mathbb{P}_{Ft}}}.
     \ar @{|->}_{
       \textred{\mathcal{E}(Fi)}
     } @<9pt> [u]
}
\end{equation*}
\caption{$F$-martingale}
\label{fig:f_martingale}
\end{figure}

\section{A Binomial Asset Pricing Model}
\label{sec:BAPM}

In this section,
we introduce a binomial asset pricing model based on the category
$\Prob$.

\subsection{Filtration $\mathcal{B}$}
\label{sec:FiltB}

First, we define a general scheme of our model by introducing a filtration
$\mathcal{B}$.

\begin{defn}{[Filtration $\mathcal{B}$]}
\label{defn:tb1}
Let
$\omega$
be the category 
whose objects are all integers starting with 
$0$
and for each pair of
integers
$m$
and
$n$
with
$m \le n$
there is a unique arrow
$
*_{m, n}
	:
m
	\to
n
$.
That is,
$\omega$
is the category corresponding to the integer set
$\mathbb{N}$
with the usual total order.
Let
$
\textred{
    \mathbf{p}
}
    :=
\{
    \textred{p_i}
\}_{i = 1, 2, \dots}
$
be an infinite sequence
of real numbers
$
p_i \in [0, 1]
$.
We define an $\omega$-filtration
$
\mathcal{B}
	:=
\mathcal{B}^{\mathbf{p}}
	: \omega \to \Prob
$
in the following way.

For
an object
$n$
of
$\omega$,
$\mathcal{B}n$
is 
a probability space
$
\bar{B}_n :=
(B_n, \Sigma_n, \mathbb{P}_n)
$
whose components are defined as follows:
\begin{enumerate}
\item
$
B_n := \{
	0, 1
\}^n
$,
the set of all binary numbers of 
$t$
digits,

\item
$
\Sigma_n
	:=
2^{B_n}
$,

\item
for
$
a
    :=
d_1 d_2 \dots d_n
    \in
B_n
$
where
$
d_i \in \{ 0, 1\}
\;
(i = 1, 2, \dots n)
$.
$
\textred{
    \mathbb{P}_n
}
    :
\Sigma_n
    \to
[0,1]
$
is the probability measure defined by
\begin{equation}
\mathbb{P}_n(\{a\})
    :=
\prod_{i=1}^n
p_i^{d_i}
(1-p_i)^{1-d_i} .
\end{equation}
\end{enumerate}

For integers
$m$
and
$n$
with
$
m < n
$,
we define
\begin{equation}
\mathcal{B}(
	*_{m, n}
)
	:=
f_{m,n}^-
	:=
(
f_m
	\circ
f_{m+1}
	\circ
\dots
	\circ
f_{n-1}
)^-
\end{equation}
where
$
f_n
	:=
(
\mathcal{B}(
	*_{n, n+1}
)
)^+
$
is a predefined null-preserving function
from
$B_{n+1}$
to
$B_{n}$.

The filtration
$\mathcal{B}$
is called
\newword{non-trivial}
if there exists
$i$
such that
$
0 < p_i < 1
$.
\end{defn}

Note that
any function from 
$B_n$
is measurable
since
$\Sigma_n$
is a powerset of 
$B_n$.

\begin{equation*}
\xymatrix@C=25 pt@R=15 pt{
\textred{\omega}
       \ar @{->}_{\textred{\mathcal{B}}} [d]
&
	0
       \ar @{->}^{i_{0}} [r]
&
	1
       \ar @{->}^{i_{1}} [r]
&
	\dots
       \ar @{->}^{i_{n-1}} [r]
&
	n
       \ar @{->}^{i_{n}} [r]
&
	n+1
       \ar @{->}^{i_{n+1}} [r]
&
	\dots
\\
\textred{\Prob}
&
	\ProbB_0
       \ar @{->}^{\textred{f^-_{0}}} [r]
&
	\ProbB_1
       \ar @{->}^{\textred{f^-_{1}}} [r]
&
	\dots
       \ar @{->}^{\textred{f^-_{n-1}}} [r]
&
	\ProbB_n
       \ar @{->}^{\textred{f^-_{n}}} [r]
&
	\ProbB_{n+1}
       \ar @{->}^{\textred{f^-_{n+1}}} [r]
&
	\dots
}
\end{equation*}

As we introduced,
the functor
$
\textred{
	\mathcal{B}
}
$
is a
generalized filtration,
representing a filtration over
the classical binomial model,
for example developed in
\cite{shreve_I}.

The classical version requires the terminal time horizon
$\textred{T}$
for determining the underlying set
$
\textred{\Omega}
	:=
\{0, 1\}^T
$
while our version does not require it
since the time variant probability spaces can evolve without any limit.
That is, 
\textblue{
our version allows unknown future elementary events,
which, we believe, shows a big philosophical difference from the Kolmogorov world.
}

\lspace

In order to see a variety of filtrations, we introduce two candidates of 
$f_n$.

\begin{defn}{[Candidates of $f_n$]}
\begin{enumerate}
\item
$
\textred{f^{full}_n}
$
\begin{equation*}
\xymatrix@C=35 pt@R=30 pt{
    B_{n+1}
        \ar @{->}^{\textred{f^{full}_n}} [r]
        \ar @{}_{ \mathrel{ \rotatebox[origin=c]{90}{$\in$} } } @<+6pt> [d]
&
	B_n
        \ar @{}_{ \mathrel{ \rotatebox[origin=c]{90}{$\in$} } } @<+6pt> [d]
\\
	d_1 \dots d_n d_{n+1}
        \ar @{|->}^{\textred{f^{full}_n}} [r]
&
	d_1 \dots d_n
}
\end{equation*}
 
\item
$
\textred{f^{drop}_n}
$
\begin{equation*}
\xymatrix@C=35 pt@R=30 pt{
    B_{n+1}
        \ar @{->}^{\textred{f^{drop}_n}} [r]
        \ar @{}_{ \mathrel{ \rotatebox[origin=c]{90}{$\in$} } } @<+6pt> [d]
&
	B_n
        \ar @{}_{ \mathrel{ \rotatebox[origin=c]{90}{$\in$} } } @<+6pt> [d]
\\
	d_1 \dots d_{n-1}, d_n d_{n+1}
        \ar @{|->}^{\textred{f^{drop}_n}} [r]
&
	d_1 \dots d_{n-1} \, \textblue{0}
}
\end{equation*}
 
\end{enumerate}
\end{defn}
The function
$
f^{drop}_n
$
can be interpreted to forget what happens at time $n$.

\lspace

Note that the function
$ f^{full}_n $
is always null-preserving
while
$ f^{drop}_n $
is null-preserving if and only if
$p_n = 0$.

\begin{exmp}{[Filtrations]}
As we mentioned in 
Definition \ref{defn:tb1},
all we need to determine the filtration is to specify
$
f_n
	:
B_{n+1}
	\to
B_{n}
$.
We have three examples of filtration
$\mathcal{B}$.
For
$j = 1, 2, \dots, n$,
\begin{enumerate}
\item
\textblue{Classical filtration}:
\begin{equation*}
\textred{f_n}
	:=
\textblue{f^{full}_n} .
\end{equation*}

\item
\textblue{Drop-$k$}:
\begin{equation*}
\textred{f_n}
	:= \begin{cases}
\textblue{f^{drop}_n}
\; \textrm{if} \;
n = k,
	\\
\textblue{f^{full}_n}
\; \textrm{if} \;
n \ne k.
\end{cases}
\end{equation*}

\item
\textblue{Elderly person}:
For
fixed numbers
$
k_0,
k_1
	\in
\mathbb{N}
$,
\begin{equation*}
\textred{f_n}
	:= \begin{cases}
\textblue{f^{drop}_n}
\; \textrm{if} \;
k_0 \le n \le T - k_1
	\\
\textblue{f^{full}_n}
\; \textrm{if} \;
0 \le n < k_0
	\, \textrm{or} \,
T - k_1 < n \le T .
\end{cases}
\end{equation*}

\end{enumerate}

\end{exmp}

\begin{prop}
For
a $\Prob$-arrow
$
\textred{f^-_n}
	:
\ProbB_{n}
	\to
\ProbB_{n+1}
$,
$
\textred{v}
	\in
L^1(\ProbB_{n+1})
$
and
$
\textred{a}
	\in
B_n
$,
\begin{equation}
\label{eq:ceBin}
\textred{
	E^{f_n^-}(v)
}(a)
\mathbb{P}_n(\{a\})
	=
\sum_{b \in \textblue{f_n^{-1}(a)}}
	v(b) 
	\mathbb{P}_{n+1}(\{b\}) .
\end{equation}
\end{prop}

Especially, 
with the classical filtration, we have
\begin{equation}
f_n^{-1}(a)
	=
(f^{full}_n)^{-1}(a)
	=
\{
	a0, a1
\} .
\end{equation}
Hence
\begin{align}
\textred{
	E^{f_n^-}(v)
}(a)
	&=
v(a0)
\frac{
	\mathbb{P}_{n+1}(\{a0\})
}{
	\mathbb{P}_{n}(\{a\})
}
	+
v(a1)
\frac{
	\mathbb{P}_{n+1}(\{a1\})
}{
	\mathbb{P}_{n}(\{a\})
}
\nonumber
	\\&=
v(a0)
(1 - p_{n+1})
	+
v(a1)
p_{n+1} .
\end{align}

\begin{defn}{[$\mathcal{B}$-Adapted Process $\xi_n$]}
For
$
n = 1, 2, \dots
$ 
define
a $\mathcal{B}$-adapted process
$
\textred{
\xi_n
}
$
by
\begin{equation*}
\xymatrix{
    B_n
        \ar @{->}^{\textred{\xi_n}} [rr]
        \ar @{}_{ \mathrel{ \rotatebox[origin=c]{90}{$\in$} } } @<+6pt> [d]
&&
    \mathbb{R}
        \ar @{}_{ \mathrel{ \rotatebox[origin=c]{90}{$\in$} } } @<+6pt> [d]
\\
	d_1 d_2 \dots d_n
        \ar @{|->}^{\textred{\xi_n}} [rr]
&&
	2 d_n - 1
}
\end{equation*}
\end{defn}


\begin{prop}
For
$
a \in B_n
$ with
$
\mathbb{P}_n(a) \ne 0
$,
\begin{align*}
\textred{E^{f_n^-}(\xi_{n+1})}
(a)
	&=
\sum_{e \in I_{n}(1, a)}
	\frac{\mathbb{P}_{n+1}(e)}{\mathbb{P}_n(a)}
-
\sum_{e \in I_{n}(0, a)}
	\frac{\mathbb{P}_{n+1}(e)}{\mathbb{P}_n(a)}
	\\&=
\#(f_n^{-1}(a)) p_{n+1}
	- 
\#I_{n}(0, a)
\end{align*}
where
\begin{equation*}
\textred{
I_{n}(j, a)
}
	:=
\{
	e \in f_n^{-1}(a)
		\mid
	(e)_{n+1} = j
\}
\end{equation*}
for 
$j = 0, 1$,
and
$
\textred{\#}A
$
denotes the cardinality of the set 
$A$.
\end{prop}

\subsection{Arbitrage Strategies}
\label{sec:ArbStr}

Now we define two instruments tradable in our market.

\begin{defn}{[Stock and Bond Processes]}
Let
$
\mu, \sigma, r
	\in
\mathbb{R}
$
be constants
such that
$
\sigma > 0
$,
$
\mu > \sigma - 1
$
and
$
r > -1
$.
\begin{enumerate}
\item
A 
\newword{stock process} 
$
\textred{S_n}
	:
B_n \to \mathbb{R}
$
 over 
$\mathcal{B}$
is defined by
\begin{equation}
S_0(\textred{\single{}})
	:= \textred{s_0},
\quad
S_{n+1}
	:=
(S_n \circ f_n)
(1 + \mu + \sigma \xi_{n+1})
\end{equation}
where
$
\textred{\single{}}
	\in
B_0
$
is the empty sequence,
and
$
s_0 > 0
$
is a positive constant.

\item
A \newword{bond process}
$
\textred{b_n}
	:
B_n \to \mathbb{R}
$
 over 
$\mathcal{B}$
is defined by
\begin{equation}
b_0(\single{}) := 1,
\quad
b_{n+1}
	:=
(b_n \circ f_n)
(1+r) .
\end{equation}

\end{enumerate}
\end{defn}

The condition
$
\mu > \sigma - 1
$
is necessary for keeping the stock price positive.

We sometimes call the triple
$
(\mathcal{B}, S, b)
$
a
\newword{market}
But, it does not mean that
the market will not contain other instruments.

\begin{prop}
For any
$
a \in B_n
$,
\begin{enumerate}
\item
$
E^{f_n^-}(S_{n+1})
	=
S_n
\big(
(1 + \mu)
E^{f_n^-}(1_{B_{n+1}})
	+
\sigma
E^{f_n^-}({\xi_{n+1}})
\big).
$

\item
$
\textred{E^{f_n^-}(1_{B_{n+1}})}
	(a)
	=
\frac{
	\mathbb{P}_{n+1}
	(f_n^{-1}(a))
}{
	\mathbb{P}_{n}
	(a)
} .
$

\item
$
b_n(a)
	=
(1+r)^n .
$
\end{enumerate}
\end{prop}

\begin{defn}{[Strategies]}
A
\newword{strategy}
is a sequence
$
(\phi, \psi)
	=
\{
(\phi_n, \psi_n)
\}_{n = 1, 2, \dots }
$,
where
\begin{equation}
\phi_n : B_{n-1} \to \mathbb{R}
\; \textrm{and} \;
\psi_n : B_{n-1} \to \mathbb{R} .
\end{equation}
Each element of the strategy
$
(\phi_n, \psi_n)
$
is called a 
\newword{portfolio}.
The 
\newword{value}
$V_n$
 of the portfolio at time $n$
is determined by:
\begin{equation}
V_n
	:=
\begin{cases}
S_0 \phi_{1}
	+
b_0 \psi_{1}
	\quad \textrm{if} \quad
n = 0
	\\
S_n 
( \phi_{n} \circ f_{n-1} )
	+
b_n 
( \psi_{n} \circ f_{n-1} ) 
	\quad \textrm{if} \quad
n = 1, 2, \dots 
\end{cases} 
\end{equation}

\end{defn}

\begin{defn}{[Gain Processes]}
A
\newword{
gain process
}
of the strategy
$
(\phi, \psi)
$
is the process
$
\{
G^{(\phi, \psi)}_n
\}_{n = 0, 1, 2, \dots }
$
defined by
\begin{equation}
G^{(\phi, \psi)}_n
	:=
\begin{cases}
	-
(
	S_0 \phi_{1}
		+
	b_0 \psi_{1}
)
\; \textrm{if} \;
n = 0
\\
(
	S_n (\phi_n \circ f_{n-1})
		+
	b_n (\psi_n \circ f_{n-1})
)
	-
(
	S_n \phi_{n+1}
		+
	b_n \psi_{n+1}
)
\; \textrm{if} \;
n = 1, 2, \dots
\end{cases}
\end{equation}
\end{defn}

\begin{lem}
\label{lem:spbp}
Let $n$
be an object of 
$\omega$ such that
\begin{equation}
S_n \phi_{n+1}
	+
b_n \psi_{n+1}
	= 0 .
\end{equation}
Then, we have
\begin{equation}
S_{n+1}
(\phi_{n+1} \circ f_n)
	+
b_{n+1}
(\psi_{n+1} \circ f_n)
	=
(\mu + \sigma \xi_{n+1} - r)
( (S_n \phi_{n+1}) \circ f_n) .
\end{equation}

\end{lem}
\begin{proof}
\begin{align*}
LHS
	=&
(S_n \circ f_n)
(1 + \mu + \sigma \xi_{n+1})
(\phi_{n+1} \circ f_n)
	+
(b_n \circ f_n)
(1 + r)
(\psi_{n+1} \circ f_n)
	\\=&
(1 + \mu + \sigma \xi_{n+1})
((S_n \phi_{n+1})
	\circ f_n)
	+
(1 + r)
((b_n \psi_{n+1}) \circ f_n)
	\\=&
(1 + \mu + \sigma \xi_{n+1})
((S_n \phi_{n+1})
	\circ f_n)
	-
(1 + r)
((S_n \phi_{n+1}) \circ f_n)
	=
RHS .
\end{align*}
\end{proof}

\begin{defn}{[Arbitrage Strategies]}
\begin{enumerate}
\item
A strategy
$
(\phi, \psi)
$
is called a
$\mathcal{B}$-\newword{arbitrage strategy}
if
$
\mathbb{P}_{n}\big(
	G^{(\phi, \psi)}_n
		\ge 0
\big) = 1
$
for every 
$n$,
and
$
\mathbb{P}_{n_0}\big(
	G^{(\phi, \psi)}_{n_0}
	> 0
\big) > 0
$
for some
$n_0$.

\item
The market is called 
\newword{non-arbitrage}
or
\textbf{NA}
if it does not allow 
$\mathcal{B}$-arbitrage strategies.

\end{enumerate}
\end{defn}

\begin{prop}
If
the market
$
(\mathcal{B}, S, b)
$
with a non-trivial filtration
$\mathcal{B}$
 is non-arbitrage,
then
$
\abs{
	\textred{\mu}
		-
	\textred{r}
}
	< 
\textred{\sigma} .
$

\end{prop}
\begin{proof}
Assuming that
$
r \le \mu - \sigma
$
or
$
r \ge \mu + \sigma
$,
we will construct an arbitrage strategy
$
(\phi, \psi)
$
by using the following algorithm.
\begin{lstlisting}[basicstyle=\ttfamily\footnotesize,frame=single]
for n = 0, 1, 2, ...:
    observe S(n) and b(n)
    if r <= mu - sigma:
        phi(n+1) > 0  # pick arbitrarily
    elsif r >= mu + sigma:
        phi(n+1) < 0  # pick arbitrarily
    psi(n+1) := -(S(n) / b(n)) phi(n+1)
    # Then, we have S(n) phi(n+1) + b(n) psi(n+1) = 0,
    # which simplifies the computation of G(n) in the following.
    if n == 0:
        G(0) := 0
    else:  # n > 1
        G(n) := s(n) (phi(n) * f(n-1)) + b(n) (psi(n) * f(n-1))
\end{lstlisting}
In the above code, 
`*'
is the function composition operator.

By Lemma \ref{lem:spbp},
we have
\begin{equation}
G^{(\phi, \psi)}_{n}
	=
(\mu + \sigma \xi_n - r)
((S_{n-1} \phi_n) \circ f_{n-1}) .
\end{equation}
So
we have
$
G^{(\phi, \psi)}_{n}
\ge 0
$
as long as 
$
r \le \mu - \sigma
$
or
$
r \ge \mu + \sigma
$.

By the way,
since our filtration is non-trivial,
there exists a number 
$n_0$
such that
$
0 < p_{n_0} < 1
$.
It is easy to check that
\begin{equation}
\mathbb{P}_{n_0}(
	G^{(\phi, \psi)}_{n_0}
	>0
) > 0 ,
\end{equation}
which concludes that
$
(\phi, \psi)
$
is an arbitrage strategy.
\end{proof}

\subsection{Risk-Neutral Filtrations}
\label{sec:RNFilt}

In this subsection,
we assume that
$
\abs{
	\textred{\mu}
		-
	\textred{r}
}
	< 
\textred{\sigma}
$.

Let us consider about the discounted stock process
\begin{equation}
\textred{S'_n}
    :=
b_n^{-1}
S_n .
\end{equation}

\begin{defn}{[Risk-neutral filtrations]}
An $\omega$-filtration
$\mathcal{C}$
is called a
\newword{risk-neutral}
filtration 
having the same shape of 
$\mathcal{B}$
if
$
U
	\circ
\mathcal{C}
	=
U
	\circ
\mathcal{B}
$
and discounted stock process becomes a
$\mathcal{C}$-martingale,
that is, for any
arrow of the form
$
*_{n, n+1}
	:
n
	\to
n+1
$
in
the category
$\omega$
\begin{equation}
\label{eq:riskNeuC}
(\mathcal{E}
 \circ
 \mathcal{C})
*_{n, n+1}
(S'_{n+1})
	=
S'_{n} .
\end{equation}

\end{defn}

We want to find
a risk-neutral filtration
$\mathcal{C}$.
Here is the shape of the filtration whose detail we will determine.

\begin{defn}{[Filtration $\mathcal{C}$]}
Let
$n$
be an object of the category
$\omega$.
\begin{enumerate}
\item
$
\textred{
	\mathbb{Q}_n
}
	:
\Sigma_n
	\to
[0, 1]
$
is a probability measure of
$
(B_n, \Sigma_n)
$,

\item
$
\textred{\ProbC_n}
	:=
(B_n, \Sigma_n, \mathbb{Q}_n)
$,

\item
$
\textred{
g_n
}
	:=
f_n
$.
\end{enumerate}
We define an $\omega$-filtration
$\mathcal{C}$
by
for 
$n \in Obj(\omega)$,
\begin{equation}
\mathcal{C}(n)
	:=
\ProbC_n,
	\quad
\mathcal{C}(*_{n, n+1})
	:=
g_n^- .
\end{equation}

\end{defn}

\begin{figure}[h]
\begin{equation*}
\xymatrix@C=25 pt@R=15 pt{
\textred{\omega}
       \ar @{->}_{\textred{\mathcal{C}}} [d]
&
	0
       \ar @{->}^{i_{0}} [r]
&
	1
       \ar @{->}^{i_{1}} [r]
&
	\dots
       \ar @{->}^{i_{n-1}} [r]
&
	n
       \ar @{->}^{i_{n}} [r]
&
	n+1
       \ar @{->}^{i_{n+1}} [r]
&
	\dots
\\
\textred{\Prob}
&
	\ProbC_0
       \ar @{->}^{\textred{g^-_{0}}} [r]
&
	\ProbC_1
       \ar @{->}^{\textred{g^-_{1}}} [r]
&
	\dots
       \ar @{->}^{\textred{g^-_{n-1}}} [r]
&
	\ProbC_n
       \ar @{->}^{\textred{g^-_{n}}} [r]
&
	\ProbC_{n+1}
       \ar @{->}^{\textred{g^-_{n+1}}} [r]
&
	\dots
}
\end{equation*}
\caption{Filtration $\mathcal{C}$}
\label{fig:filtration_c}
\end{figure}
With these notations, 
(\ref{eq:riskNeuC})
is reduced to the form
\begin{equation}
\label{eq:riskNCE}
E^{g^-_{n}}
(S'_{n+1})
	=
S'_{n} .
\end{equation}

\begin{thm}
\label{thm:DSP}
A process
$S_n'$
is a
\textblue{
$\mathcal{C}$-martingale
},
that is,
for
$
n \in \mathbb{N}
$,
$
E^{g_n^-}(S'_{n+1})
	=
S'_n
$
if and only if
for all
$
n \in \mathbb{N}
$
and
$
a \in B_n
$,
\begin{equation}
\label{eq:thmDSP}
\mathbb{Q}_n(
	\{a\}
)
	=
c_1 \, \mathbb{Q}_{n+1}(I_{n}(1,a))
	+
c_0 \, \mathbb{Q}_{n+1}(I_{n}(0,a))
\end{equation}
where
for 
$j = 0, 1$
\begin{equation}
\textred{
I_{n}(j,a)
}
	:=
\{
	e \in f_n^{-1}(a)
		\mid
	(e)_{n+1} = j
\}
\end{equation}
and
\begin{equation}
\textred{c_1}
	:=
\frac{1 + \mu + \sigma}{1 + r},
\quad
\textred{c_0}
	:=
\frac{1 + \mu - \sigma}{1 + r}.
\end{equation}

\end{thm}
\begin{proof}
For
$a \in B_n$
\vspace{-0.3\baselineskip}
\begin{align*}
&
\textblue{
	S'_n(a)
}
\mathbb{Q}_n(\{a\})
	\textblue{=}
\textblue{
	E^{g_n^-}(S'_{n+1})(a)
}
\mathbb{Q}_n(\{a\})
	\\=&
\sum_{
	e \in f_n^{-1}(a)
}
	S'_{n+1}(e)
	\mathbb{Q}_{n+1}(\{e\})
	\\=&
\sum_{
	e \in f_n^{-1}(a)
}
	b_{n+1}^{-1}(e)
	(S_{n} \circ f_n)(e)
	(1 + \mu + \sigma \xi_{n+1}(e))
	\mathbb{Q}_{n+1}(\{e\})
	\\=&
\sum_{
	e \in f_n^{-1}(a)
}
	(1+r)^{-(n+1)}
	S_{n} (a)
	(1 + \mu + \sigma \xi_{n+1}(e))
	\mathbb{Q}_{n+1}(\{e\})
	\\
	=&
S'_{n} (a)
\sum_{
	e \in f_n^{-1}(a)
}
	\frac{
		1 + \mu + \sigma \xi_{n+1}(e)
	}{
		1+r
	}
	\mathbb{Q}_{n+1}(\{e\}) .
\end{align*}
\vspace{-0.5\baselineskip}
\noindent
\textblue{if and only if}
\begin{align*}
\mathbb{Q}_n(\{a\})
	=&
\sum_{
	e \in I_n(1, a)
}
	\!\!
	\frac{
		1 + \mu + \sigma
	}{
		1+r
	}
	\mathbb{Q}_{n+1}(\{e\})
+
\sum_{
	e \in I_n(0, a)
}
	\!\!
	\frac{
		1 + \mu - \sigma
	}{
		1+r
	}
	\mathbb{Q}_{n+1}(\{e\})
	\\=&
c_1 \,
	\mathbb{Q}_{n+1}(I_n(1, a))
+
c_0 \,
	\mathbb{Q}_{n+1}(I_n(0, a)) .
\end{align*}
\end{proof}

In order to determine more detail of 
$\mathcal{C}$,
we need the following condition for
$\mathbb{Q}_n$.

\begin{prop}
\label{prop:Qcond}
The following conditions for
$
\mathbb{Q}_n
$
are equivalent.
\begin{enumerate}
\item
for all
$n \in \mathbb{N}$,
$
a \in
B_n
$,
\begin{equation}
\mathbb{Q}_{n+1}(\{a0, a1\})
	=
\mathbb{Q}_n(\{a\})
\end{equation}

\item
for all
$n \in \mathbb{N}$,
$f^{full}_n$
is measure-preserving
w.r.t.
$\mathbb{Q}_n$,
that is,
\begin{equation}
\mathbb{Q}_{n}
	=
\mathbb{Q}_{n+1}
	\circ
(f^{full}_n)^{-1} .
\end{equation}

\item
there exists
a sequence of functions
$
\{
	q_k : B_k \to [0, 1]
\}_{k = 1, 2, \dots}
$
such that
for all
$n = 1, 2, \dots$
and
$d_j = 0, 1$,
\begin{equation}
\mathbb{Q}_{n}(\{d_1 d_2 \dots d_n\})
	=
\prod_{k=1}^n
	q_{k}(d_1 d_2 \dots d_k)
\end{equation}
such that for every
$
a \in B_{n-1}
$,
$
q_n(a 0) + q_n(a 1) = 1 .
$

\end{enumerate}
\end{prop}

In the following discussion, 
we assume the following assumption
which is the condition
(2) of Proposition \ref{prop:Qcond}.

\begin{asm}
\label{asm:fullmp}
For all
$n \in \mathbb{N}$,
$f^{full}_n$
is measure-preserving
w.r.t.
$\mathbb{Q}_n$.

\end{asm}

By
Assumption \ref{asm:fullmp}
and
(3) of Proposition \ref{prop:Qcond},
we have
\begin{equation*}
\mathbb{Q}_{n+1} (\{d_1 d_2 \dots d_n d_{n+1}\})
	=
\mathbb{Q}_n (\{d_1 d_2 \dots d_n\})
q_{n+1}(
	d_1 d_2 \dots d_{n+1}
) .
\end{equation*}

\lspace

In the rest of this subsection,
we will investigate the shape of
$\mathbb{Q}_n$
under the assumption that
$S'_n$
is
$\mathcal{C}$-martingale.

\subsubsection{Classical Filtration}
\label{sec:classical_filtration}

First, we prepare a lemma for for the proof of the following propositions.

\begin{lem}
\label{lem:qbin}
If
$
1
	=
c_1 x + c_0 (1-x)
$,
then
\begin{equation}
x
	=
\frac{1}{2}
	+
\frac{
	r - \mu
}{
	2 \sigma
}
\quad \textrm{and} \quad
1 - x
	=
\frac{1}{2}
	-
\frac{
	r - \mu
}{
	2 \sigma
} .
\end{equation}
\end{lem}

\begin{prop}
For a fixed
$n \in \mathbb{N}$,
assume that 
$
f_n
	=
f^{full}_n
$.
Then for
$
a \in B_n
$
with
$
\mathbb{Q}_n(\{a\})
	\ne 0
$,
we have
\begin{align*}
q_{n+1}(a1)
	&=
\frac{1}{2}
	+
\frac{r - \mu}{2 \sigma},
	\\
q_{n+1}(a0)
	&=
\frac{1}{2}
	-
\frac{r - \mu}{2 \sigma} .
\end{align*}

\end{prop}

Note that the resulting probability depends neither
on $a$
nor on $n$.
\begin{proof}
By observing the following diagram
\begin{equation*}
\xymatrix@C=37 pt@R=0 pt{
&
	a1
\\
\\
	a
		\ar @{-} [ruu]
		\ar @{-} [rdd]
\\
\\
&
	a0
\\\\
    B_n
        \ar @{}_{ \mathrel{ \rotatebox[origin=c]{90}{$\in$} } } @<+6pt> [dd]
&
    B_{n+1}
        \ar @{->}_{f^{full}_n} [l]
        \ar @{}_{ \mathrel{ \rotatebox[origin=c]{90}{$\in$} } } @<+6pt> [dd]
\\\\
	a
&
	a d_{n+1}
        \ar @{|->}_{f^{full}_n} [l]
}
\end{equation*}
we have
\begin{align*}
(f^{full}_n)^{-1}(a)
	&=
\{ a0, a1 \}
\\
I_n(1, a) &= \{ a1 \},
\\
I_n(0, a) &= \{ a0 \}
\end{align*}
By (\ref{eq:thmDSP})
\begin{align*}
\mathbb{Q}_n(\{a\})
	&=
c_1 \mathbb{Q}_{n+1}(I_n(1,a))
	+
c_0 \mathbb{Q}_{n+1}(I_n(0,a))
	\\&=
c_1 \mathbb{Q}_{n+1}(\{a1\})
	+
c_0 \mathbb{Q}_{n+1}(\{a0\})
\end{align*}
Now since
\begin{equation*}
\mathbb{Q}_{n+1}(\{a d_{n+1}\})
	=
\mathbb{Q}_{n}(\{a\})
q_{n+1}(a d_{n+1})
\end{equation*}
and
$
\mathbb{Q}_{n}(\{a\})
	\ne 0
$,
we have
\begin{equation*}
1 =
c_1 q_{n+1}(a1)
	+
c_0 q_{n+1}(a0) .
\end{equation*}
Hence by 
Lemma \ref{lem:qbin},
we have
\begin{equation*}
q_{n+1}(a1)
	=
\frac{1}{2}
	+
\frac{r - \mu}{2 \sigma},
\quad
q_{n+1}(a0)
	=
\frac{1}{2}
	-
\frac{r - \mu}{2 \sigma} .
\end{equation*}

\end{proof}

\begin{cor}
If
$\mathcal{B}$
is the classical filtration,
then for any 
$n \in \mathbb{N}$
and
$a \in B_n$
we have
\begin{equation}
\mathbb{Q}_n(a)
	=
\big(
	\frac{1}{2}
		+
	\frac{r - \mu}{2 \sigma}
\big)^{n(1, a)}
\big(
	\frac{1}{2}
		-
	\frac{r - \mu}{2 \sigma}
\big)^{n(0, a)}
\end{equation}
where
\begin{equation}
\textred{
	n(j, a)
}
	:= \#\{
	k
\mid
	(a)_k = j	
\} .
\end{equation}
\end{cor}

\subsubsection{Drop-$k$ Filtration}
\label{sec:dropk_filtration}

\begin{prop}
For a fixed
$n (= 1, 2, \dots)$,
assume that 
$
f_n
	=
f^{drop}_n
$.
Then for
$
a \in B_{n-1}
$
with
$
\mathbb{Q}_{n-1}(\{a\})
	\ne 0
$,
we have
\begin{align*}
q_{n}(a1)
	&=
0,
\\
q_{n}(a0)
	&=
1,
\\
q_{n+1}(a01)
	&=
\frac{1}{2}
	+
\frac{r - \mu}{2 \sigma},
	\\
q_{n+1}(a00)
	&=
\frac{1}{2}
	-
\frac{r - \mu}{2 \sigma} .
\end{align*}

\end{prop}
\begin{proof}
By observing the following diagram
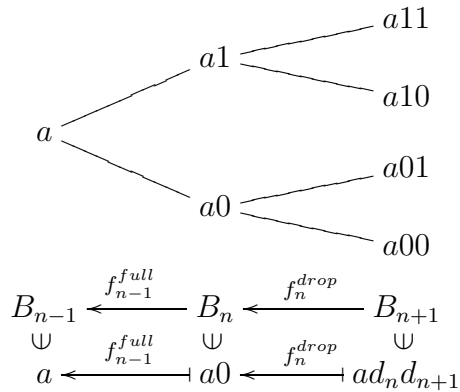
\begin{equation*}
\xymatrix@C=37 pt@R=0 pt{
&&
	a11
\\
&
	a1
		\ar @{-} [ru]
		\ar @{-} [rd]
\\
&&
	a10
\\
	a
		\ar @{-} [ruu]
		\ar @{-} [rdd]
\\
&&
	a01
\\
&
	a0
		\ar @{-} [ru]
		\ar @{-} [rd]
\\
&&
	a00
\\\\
	B_{n-1}
        \ar @{}_{ \mathrel{ \rotatebox[origin=c]{90}{$\in$} } } @<+6pt> [dd]
&
    B_n
        \ar @{->}_{f^{full}_{n-1}} [l]
        \ar @{}_{ \mathrel{ \rotatebox[origin=c]{90}{$\in$} } } @<+6pt> [dd]
&
    B_{n+1}
        \ar @{->}_{f^{drop}_n} [l]
        \ar @{}_{ \mathrel{ \rotatebox[origin=c]{90}{$\in$} } } @<+6pt> [dd]
\\\\
	a
&
	a\textblue{0}
        \ar @{|->}_{f^{full}_{n-1}} [l]
&
	a d_n d_{n+1}
        \ar @{|->}_{f^{drop}_n} [l]
}
\end{equation*}
we have
\begin{align*}
(f^{drop}_n)^{-1}(a1) &= \emptyset
\\
I_n(1, a1) &=
I_n(0, a1) = \emptyset
\\
(f^{drop}_n)^{-1}(a0) 
	&= 
\{ a00, a01, a10, a11 \}
\\
I_n(1, a0) &= \{ a01, a11 \}
\\
I_n(0, a0) &= \{ a00, a10 \}
\end{align*}
By (\ref{eq:thmDSP})
\begin{equation*}
\mathbb{Q}_n(\{a1\})
	=
c_1 \mathbb{Q}_{n+1}(I_n(1,a1))
	+
c_0 \mathbb{Q}_{n+1}(I_n(0,a1))
	=
0 .
\end{equation*}
Now since
$
\mathbb{Q}_{n}(\{a d_{n}\})
	=
\mathbb{Q}_{n-1}(\{a\})
q_{n}(a d_{n}) 
$
and
$
\mathbb{Q}_{n-1}(\{a\})
	\ne 0
$,
we have
\begin{equation*}
q_{n}(a 1) = 0,
\quad
q_{n}(a 0) 
	=
1 - q_{n}(a 1) 
	=
1.
\end{equation*}
Next, again by (\ref{eq:thmDSP})
\begin{align*}
\mathbb{Q}_n(\{a0\})
	&=
c_1 \mathbb{Q}_{n+1}(I_n(1,a0))
	+
c_0 \mathbb{Q}_{n+1}(I_n(0,a0))
	\\&=
c_1 \big(
	\mathbb{Q}_{n+1}(\{a01\})
		+
	\mathbb{Q}_{n+1}(\{a11\})
\big)
	\\&+
c_0 \big(
	\mathbb{Q}_{n+1}(\{a00\})
		+
	\mathbb{Q}_{n+1}(\{a10\})
\big)
\end{align*}
By dividing both hands by
$
	\mathbb{Q}_{n-1}(\{a\})
\ne
	0
$,
\begin{align*}
q_n(a0)
	&=
c_1 \big(
	q_n(a0) q_{n+1}(a01)
		+
	q_n(a1) q_{n+1}(a11)
\big)
	\\&+
c_0 \big(
	q_n(a0) q_{n+1}(a00)
		+
	q_n(a1) q_{n+1}(a10)
\big)
\end{align*}
Then, since
$
q_{n}(a 1) = 0
$
and
$
q_{n}(a 0) = 1
$,
\begin{equation*}
1
	=
c_1 q_{n+1}(a01)
	+
c_0 q_{n+1}(a00) .
\end{equation*}
Hence, by
Lemma \ref{lem:qbin},
we have
\begin{equation*}
q_{n+1}(a01)
	=
\frac{1}{2}
	+
\frac{r - \mu}{2 \sigma},
	\quad
q_{n+1}(a00)
	=
\frac{1}{2}
	-
\frac{r - \mu}{2 \sigma} .
\end{equation*}

\end{proof}

We have to check that
both
$f^{full}_n$
and
$f^{drop}_n$
are null-preserving w.r.t. 
$\mathbb{Q}_n$.

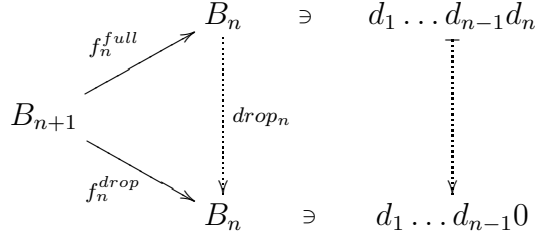
\begin{equation*}
\xymatrix@C=40 pt@R=20 pt{
&
	B_n
		\ar @{.>}^{\textred{drop}_n} [dd]
     		\ar @{}^-{\ni} @<-6pt> [r]
&
	d_1 \dots d_{n-1} \textred{d_n}
		\ar @{|.>} [dd]
\\
	B_{n+1}
		\ar @{->}^{f^{full}_n} [ru]
		\ar @{->}_{f^{drop}_n} [rd]
\\
&
	B_n
     		\ar @{}^-{\ni} @<-6pt> [r]
&
	d_1 \dots d_{n-1} \textred{0}
}
\end{equation*}

If
$
\mathbb{Q}_n(
	d_1 \dots d_{n-1} \textred{1}
)
	\ne
0
$,
then
$
\textred{drop}_n
$
is null-preserving,
and so is 
$f^{drop}_n$
since
$f^{full}_n$
is measure-preserving.

\begin{figure}[h]
\begin{equation*}
\xymatrix@C=37 pt@R=0 pt{
&&
	\textred{a11}
\\
&
	\textred{a1}
		\ar @{-} [ru]
		\ar @{-} [rd]
\\
&&
	\textred{a10}
\\
	a
		\ar @{-} [ruu]
		\ar @{-} [rdd]
\\
&&
	a01
\\
&
	a0
		\ar @{-} [ru]
		\ar @{-} [rd]
\\
&&
	a00
\\\\
	B_{n-1}
&
    B_n
        \ar @{->}_{f^{full}_{n-1}} [l]
&
    B_{n+1}
        \ar @{->}_{f^{drop}_n} [l]
}
\end{equation*}
\caption{$f_n^{drop}$}
\label{fig:fn_drop}
\end{figure}

\begin{rem}
We have the following remarks for 
Figure \ref{fig:fn_drop}.
\begin{enumerate}
\item
%
Since the agent evaluates stock and bond along the function $f_{n}^{drop}$,
she can recognise only the nodes $a0$, $a01$ and $a00$ 
and can not recognise the nodes $a1$, $a11$ and $a10$. 
We interpret these nodes $a1$, $a11$ and $a10$ as invisible. 

\item
The values
$
q_{n+1}(a11)
	\in [0, 1]
$
can be arbitrarily selected,
and
$q_{n+1}(a10)$
is computed by
$
1 - q_{n+1}(a10)
$.
That is, the probability measure
$\mathbb{Q}_{n+1}$
is not determined uniquely,
so is not the risk-neutral filtration
$\mathcal{C}$.

\item
The probability measure
$\mathbb{Q}_{n}$
is not equivalent to the original measure
$\mathbb{P}_{n}$.
Therefore, it is not an EMM.

\end{enumerate}

\end{rem}

\begin{figure}[h]
\begin{equation*}
\xymatrix@C=37 pt@R=0 pt{
&&&
	\textred{a111}
\\
&&
	\textred{a11}
		\ar @{-} [ru]
		\ar @{-} [rd]
\\
&&&
	\textred{a110}
\\
&
	a1
		\ar @{-} [ruu]
		\ar @{-} [rdd]
\\
&&&
	a101
\\
&&
	a10
		\ar @{-} [ru]
		\ar @{-} [rd]
\\
&&&
	a100
\\
	a
		\ar @{-} [ruuuu]
		\ar @{-} [rdddd]
\\
&&&
	\textred{a011}
\\
&&
	\textred{a01}
		\ar @{-} [ru]
		\ar @{-} [rd]
\\
&&&
	\textred{a010}
\\
&
	a0
		\ar @{-} [ruu]
		\ar @{-} [rdd]
\\
&&&
	a001
\\
&&
	a00
		\ar @{-} [ru]
		\ar @{-} [rd]
\\
&&&
	a000
\\\\
	B_{k-2}
&
	B_{k-1}
        \ar @{->}_{f^{full}_{k-2}} [l]
&
    B_k
        \ar @{->}_{f^{full}_{k-1}} [l]
&
    B_{k+1}
        \ar @{->}_{\textblue{f^{drop}_k}} [l]
}
\end{equation*}
\caption{drop-$k$ filtration}
\label{fig:dropk_filtration}
\end{figure}
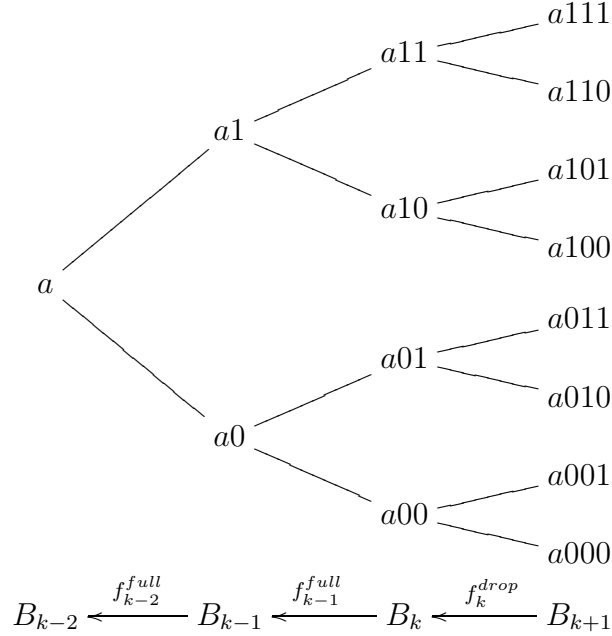

\begin{rem}
Let
$
\textred{\mathcal{C}}
 : \omega \to \Prob
$
be a risk-neutral filtration,
and
$
\textred{Y}
 : B_{T} \to \mathbb{R}
$
be a \textblue{payoff} at time 
$
\textred{T}
$.
Then, 
for the agent who has 
a drop-$k$ filtration
as her subjective filtration,
the price of 
$Y$
at time 
$
\textred{n}
$
with a unique arrow
$
\textred{i_n} : n \to T
$
is given by
\begin{equation*}
\textred{Y_n}
	:=
E^{\mathcal{C}i_n}(
	b_T^{-1} Y
) .
\end{equation*}

\begin{equation*}
\xymatrix@C=15 pt@R=30 pt{
	n
		\ar @{->}_{i_n} [d]
		\ar @{|->}^{\mathcal{C}} [rr]
  &&
	\ProbC_n
		\ar @{->}_{\mathcal{C}{i_n}} [d]
		\ar @{|->}^{\mathcal{E}} [rr]
  &&
	\mathcal{E}(\ProbC_n)
		\ar @{}^-{:=} @<-6pt> [r]
  &
    L^{1}(\ProbC_n)
     \ar @{}^-{\ni} @<-6pt> [r]
  &
	Y_n = E^{\mathcal{C}{i_n}}(b_T^{-1} Y)
\\
    T
		\ar @{|->}^{\mathcal{C}} [rr]
  &&
	\ProbC_T
    	\ar @{|->}^{\mathcal{E}} [rr]
  &&
	\mathcal{E}(\ProbC_T)
		\ar @{}^-{:=} @<-6pt> [r]
		\ar @{->}_{
			\textred{\mathcal{E}({\mathcal{C}}{i_n})}
		} [u]
  &
    L^{1}(\ProbC_T)
     \ar @{}^-{\ni} @<-6pt> [r]
  &
	b_T^{-1} Y
     \ar @{|->}_{
       \textred{\mathcal{E}(\mathcal{C}{i_n})}
     } @<9pt> [u]
}
\end{equation*}

\end{rem}

For
$a \in B_{n-2}$,
you can see in 
Figure \ref{fig:val_fndrop}
that
at time
$n-1$
the value of
$Y_n(a1)$
is discarded and use only the value of
$Y_n(a0)$
for computing 
$Y_{n-1}(a)$.

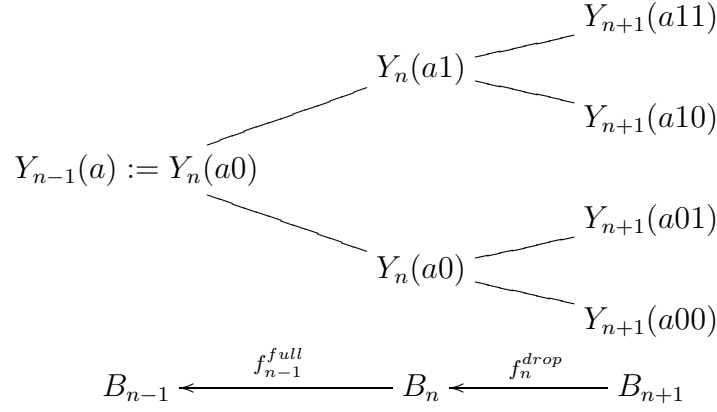
\begin{figure}[h]
\begin{equation*}
\xymatrix@C=37 pt@R=0 pt{
&&
	\textred{Y_{n+1}(a11)}
\\
&
	\textred{Y_n(a1)}
		\ar @{-} [ru]
		\ar @{-} [rd]
\\
&&
	\textred{Y_{n+1}(a10)}
\\
	\textblue{Y_{n-1}(a) := Y_n(a0)}
		\ar @{-} [ruu]
		\ar @{-} [rdd]
\\
&&
	Y_{n+1}(a01)
\\
&
	Y_n(a0)
		\ar @{-} [ru]
		\ar @{-} [rd]
\\
&&
	Y_{n+1}(a00)
\\\\
	B_{n-1}
&
    B_n
        \ar @{->}_{f^{full}_{n-1}} [l]
&
    B_{n+1}
        \ar @{->}_{f^{drop}_n} [l]
}
\end{equation*}
\caption{Valuation along $f_n^{drop}$}
\label{fig:val_fndrop}
\end{figure}


\subsubsection{Replication Strategies}

Let us investigate the situation where
a given strategy
$
(\phi, \psi)
$
becomes a replication strategy of
the payoff 
$Y$
at time
$T$.

\begin{defn}{[Self-Financial Strategies]}
A
\newword{self-financial}
strategy
is 
a strategy
$
(\phi, \psi)
$
satisfying
\begin{equation}
S_n \phi_{n+1}
	+
b_n \psi_{n+1}
	=
V_n
\end{equation}
for
every
$
n = 1, 2, \dots
$.

\end{defn}

For a self-financial strategy
$
(
\phi_n, \psi_n
)_{n = 1, 2, \dots}
$,
we have:
\begin{align*}
V_{n+1}
	&=
S_{n+1} (\phi_{n+1} \circ f_n)
	+
b_{n+1} (\psi_{n+1} \circ f_n)
	\\&=
(S_n \circ f_n)(1 + \mu + \sigma \xi_{n+1})
(\phi_{n+1} \circ f_n)
	+
b_{n+1} 
(
b_n^{-1}
(V_n - S_n \phi_{n+1})
\circ f_n)
	\\&=
(1 + \mu + \sigma \xi_{n+1})
((S_n \phi_{n+1}) \circ f_n)
	+
(1+r)
(
(V_n - S_n \phi_{n+1})
\circ f_n
)
	\\&=
(\mu - r + \sigma \xi_{n+1})
((S_n \phi_{n+1}) \circ f_n)
	+
(1+r)
(
V_n 
\circ f_n
) .
\end{align*}
Therefore,
for
$a \in B_n$
and
$
d_{n+1}
	\in
\{
0, 1
\}
$,
\begin{equation}
\label{eq:vnPreEq}
V_{n+1}
(a d_{n+1})
	=
(\mu - r + \sigma (2 d_{n+1} - 1))
S_n(b_n) 
\phi_{n+1}(b_n)
	+
(1+r)
V_n(b_n)
\end{equation}
where
\begin{equation}
b_n
	:=
f_n(a d_{n+1}) .
\end{equation}

Now let us assume that there exists a function
$
g_n 
	:
B_n
	\to
B_n
$
such that
$
f_n =
	g_n
\circ
	f^{full}_n
$ .
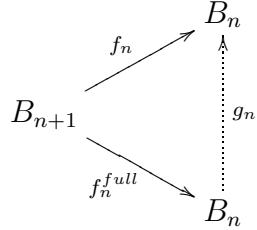
\begin{equation*}
\xymatrix@C=40 pt@R=20 pt{
&
	B_n
\\
	B_{n+1}
		\ar @{->}^{f_n} [ru]
		\ar @{->}_{f^{full}_n} [rd]
\\
&
	B_n
		\ar @{.>}_{\textred{g_n}} [uu]
}
\end{equation*}
Then
$
f_n(a d_{n+1})
	=
g_n(a)
$
for every
$
a \in B_n
$
and
$
d_{n+1}
	\in
\{ 0, 1 \}
$.
So the equation
(\ref{eq:vnPreEq})
becomes
\begin{equation}
\label{eq:vnPreMEq}
V_{n+1}
(a d_{n+1})
	=
(\mu - r + \sigma (2 d_{n+1} - 1))
S_n(g_n(a)) 
\phi_{n+1}(g_n(a))
	+
(1+r)
V_n(g_n(a)) .
\end{equation}
Hence, we have:
\begin{align}
\label{eq:phiN}
\phi_{n+1}(g_n(a))
	&=
\frac{
	V_{n+1}(a1)
		-
	V_{n+1}(a0)
}{
	2 \sigma S_n(g_n(a))
}
\\
\label{eq:VN}
V_n(g_n(a))
	&=
\frac{
	(\sigma - \mu + r)
	V_{n+1}(a1)
		+
	(\sigma + \mu - r)
	V_{n+1}(a0)
}{
	2 \sigma (1+r)
} .
\end{align}

Therefore,
we can determine the appropriate strategy
$
(\phi_{n+1}, \psi_{n+1})
$
on 
$
g_n(B_n)
	\subset
B_n
$
by (\ref{eq:phiN}).
We actually do not care the values of
$
(\phi_{n+1}, \psi_{n+1})
$
on
$
B_n \setminus g_n(B_n)
$.

For example,
in the case of
$
f_n
	=
f^{drop}_n
$,
the function
$
g_n
	:
B_n \to B_n
$
satisfies
\begin{equation}
g_n(d_1 \dots d_{n-1} d_n)
	=
d_1 \dots d_{n-1} 0
\end{equation}
for all
$
d_1 \dots d_{n-1} d_n
	\in
B_n
$.
Looking at
Figure \ref{fig:val_fndrop},
values in
the region
$
B_n \setminus g_n(B_n)
$
are not necessary for computing
$
Y_{n-1}(a)
$.
Hence,
determining
the values of
$
(\phi_{n+1}, \psi_{n+1})
$
in
$
g_n(B_n)
$
is enough for making the practical valuation.


\section{Concluding Remarks}

We formulated an infinitely growing sequence of binomial probability spaces
in the category
$\Prob$.
%
We gave some concrete (possibly distorted) filtrations.
%
We determined the shape of the risk-neutral filtrations to the above examples.
%
We showed the valuations of claims given at time $T$
through the distorted filtrations,
and provided a replication strategy implementing the valuation.




\nocite{maclane1997}
\nocite{CK2012DM}

\def \EmbedBib  {1}

\if \EmbedBib  0
\bibliographystyle{apalike}
\bibliography{../../taka_e}
\else 

\fi 


\end{document}